\documentclass[a4paper]{article}
\usepackage{amsmath}
\usepackage{amssymb}
\usepackage{dsfont} 
\usepackage{url}
\usepackage{tikz}
\usepackage{microtype}
\usepackage{blkarray}
\usepackage{bbm}
\usepackage{algpseudocode}
\usepackage{algorithm}
\usepackage{pgfplots}
\usetikzlibrary{arrows}
\tikzset{>=stealth'}

\newcommand{\comment}[1]{}

\newcommand{\mx}[1]{\mathbf{#1}}
\newcommand{\1}{\mathds{1}}

\newcommand{\ve}[1]{\underline{\smash{#1}}}

\newcommand{\ii}{\mathbbm{i}}

\newtheorem{theorem}{Theorem}
\newtheorem{definition}{Definition}
\newtheorem{corollary}{Corollary}
\newenvironment{proof}{\paragraph{Proof:}}{\hfill$\square$}

\title{High order concentrated non-negative matrix-exponential functions}

\author{Gabor Horvath$^1$ \and Illes Horvath$^2$ \and Miklos Telek$^{1,2}$}

\date{ $^1$ Budapest University of Technology and Economics, Dept. of Networked Systems and Services, Magyar Tudosok Korutja 2, 1117 Budapest, Hungary (ghorvath@hit.bme.hu,~telek@hit.bme.hu) \\
$^2$ MTA-BME Information Systems Research Group, Magyar Tudosok Korutja 2, 1117 Budapest, Hungary (horvath.illes.antal@gmail.com,~telek@hit.bme.hu)
}

\begin{document}
\maketitle


\begin{abstract}
Highly concentrated functions play an important role in many research fields including control system analysis and physics, and they turned out to be the key idea behind inverse Laplace transform methods as well.

This paper uses the matrix-exponential family of functions to create highly concentrated functions, whose squared coefficient of variation (SCV) is very low. In the field of stochastic modeling, matrix-exponential functions have been used for decades. They have many advantages: they are easy to manipulate, always non-negative, and integrals involving matrix-exponential functions often have closed-form solutions.  For the time being there is no symbolic construction available to obtain the most concentrated matrix-exponential functions, and the numerical optimization-based approach has many pitfalls, too.

In this paper, we present a numerical optimization-based procedure to construct highly concentrated matrix-exponential functions. To make the objective function explicit and easy to evaluate we introduce and use a new representation called hyper-trigonometric representation. This representation makes it possible to achieve very low SCV. 
\end{abstract}

\section{Introduction}

The class of finite dimensional matrix exponential functions is of special interest in control system analysis \cite{DeSc00,DGJM09}. In Laplace transform domain they form the class of rational functions \cite{lipsky2008queueing},
and there are plenty of Laplace transform domain numerical procedures that are applicable only for the set of rational functions \cite{LTRF}. 

Additionally, in the analysis of stochastic models matrix exponential distributions \cite{AsBl99} extend the analysis of memoryless systems such that the extended systems remain tractable with efficient numerical procedures \cite{BeNi07,[BUCH10c],bladt2017matrix}. 

Extreme members of finite dimensional matrix exponential functions are of special interest in many numerical procedures, for which numeric inverse Laplace transformation is a striking example \cite{[IHORV17a]}. 

In this paper, we cope with the numerical burden of computing higher order
non-negative matrix-exponential function, $f(t)$, which is concentrated in the sense that its squared coefficient of variation (SCV),
defined as $\frac{\mu_0 \mu_2}{\mu_1^2}-1$, where $\mu_i=\int_{t=0}^{\infty}t^i f(t) dt$, is low. 
Preliminary results are available in \cite{elteto2006minimal} up to $N=17$
and in \cite{horvath2016concentrated} up to $N=47$. These preliminary results indicate that the minimal SCV of the non-negative order $N$ class tends to be less than $2/N^2$. 
In this work, we propose numerical procedures by which much higher order concentrated non-negative matrix-exponential functions can be computed and based on that we refine the dependence of the minimal SCV on the order. 

The connections of control theory and stochastic models have been recognized previously \cite{[BUCH10a],DeSc00}. The authors' background is more in stochastic analysis, and consequently, we apply a stochastic analysis related notations and terminology, but the notions \emph{matrix-exponential distribution} and \emph{non-negative matrix-exponential functions} can be freely interchanged along the text. 

The rest of the paper is organized as follows.
Section \ref{sec:me} introduces the matrix-exponential functions, presents some of its properties and discusses the difficulty of finding concentrated functions. The main contribution of the paper is described in Section \ref{sec:scv}, where a new representation transformation method is presented to enable the efficient computation of the SCV. The details and the results of the optimization procedure are provided in Section \ref{sec:opt}. A new, efficient heuristic is introduced and evaluated in Section \ref{sec:heuristic}. Finally, Section \ref{sec:concl} concludes the paper.

\section{Concentrated non-negative ME functions}\label{sec:me}

In this work we focus on the following set of functions. 

\begin{definition}
Order $N$ \emph{matrix-exponential functions} (referred to as ME($N$)) are given by
\begin{align}
\label{eq:ft}
f(t)=\ve{\alpha} e^{\mx{A}t} (-\mx{A})\1,
\end{align}
where $\ve{\alpha}$ is a real row vector of size $N$, $\mx{A}$ is a real matrix of size $N\times N$ matrix and $\1$ is the column vector of ones of size $N$.
\end{definition}

A more general definition, where the column vector on the right is not fixed to $\1$, is also used frequently in the literature, but that generalization does not enlarge the set of functions (as proven in \cite{lipsky2008queueing,bladt2017matrix}). 

According to \eqref{eq:ft},  vector $\ve{\alpha}$ and matrix $\mx{A}$ define a matrix exponential function, which we refer to as \emph{matrix representation} in the sequel.
The matrix representation is not unique. Applying a similarity transformation with a non-singular matrix $\mx{T}$, for which $\mx{T}\1=\1$ holds, we obtain vector $\ve{\beta}=\ve{\alpha}\mx{T}^{-1}$ and matrix $\mx{B}=\mx{T}\mx{A}\mx{T}^{-1}$, that define the same matrix exponential function, since
\begin{align*}
\ve{\beta} e^{\mx{B}t} (-\mx{B})\1&= \ve{\alpha}\mx{T}^{-1} e^{\mx{T}\mx{A}\mx{T}^{-1}t} (-\mx{T}\mx{A}\mx{T}^{-1})\1
\\&= \ve{\alpha}\mx{T}^{-1}\mx{T} e^{\mx{A}t}\mx{T}^{-1}\mx{T} (-\mx{A})\mx{T}^{-1}\1= 
\ve{\alpha} e^{\mx{A}t} (-\mx{A})\1=f(t).
\end{align*}

We are interested in the non-negative concentrated members of the ME($N$) class. A non-negative function on $\Re^+$ is said to be concentrated when its squared coefficient of variation 
\begin{align}
SCV(f(t))=\frac{\mu_0 \mu_2}{\mu_1^2}-1,
\label{eq:SCV}
\end{align}
is low. In \eqref{eq:SCV}, $\mu_i$ denotes the $i$th moment, defined by $\mu_i=\int_{t=0}^{\infty}t^i f(t) dt$ for $i=0,1,2$. 

Although matrix-exponential functions have been used for many decades, there are still many questions open regarding their properties.

Such an important question is how to decide efficiently if a matrix-exponential function is non-negative for $\forall t>0$. In general, $f(t)\geq 0, \forall t>0$ does not necessarily hold for given $(\ve{\alpha},\mx{A})$ parameters, unless it has been constructed to be always non-negative. In this paper, we are going to restrict our attention to such a special construction, to the exponential-cosine square functions.

Another important, still open question is which subclass of the general ME($N$) class provides the most concentrated functions. For the special case when the elements of $\ve{\alpha}$ are probabilities, and the off-diagonals of $\mx{A}$ are non-negative, it is known that the Erlang distribution has the minimal SCV, and it is $1/N$ \cite{[ALDO87]}. For the most concentrated (general) non-negative ME($N$) functions only conjectures are available for $N\geq 3$ \cite{elteto2006minimal}. According to the current conjecture for odd $N$, the most concentrated non-negative ME($N$) belongs to a special subset of ME($N$) given by the definition below, which is the main interest for us in this work. 

\begin{definition}
The set of \emph{exponential cosine-square} functions of order $n$
has the form
\begin{align}
f^+(t) = e^{-t} \prod_{i=1}^{n} \cos^2\left(\frac{\omega t - \phi_i}{2}\right).\label{eq:cossquare}
\end{align}
\end{definition}

An exponential cosine-square function is defined by $n+1$ parameters: $\omega$ and $\phi_i$ for $i=1,\ldots,n$.
An exponential cosine-square function is a matrix exponential function.
Although the representation in \eqref{eq:cossquare}, which we refer to as \emph{cosine-square representation}, is not a matrix representation, in Section \ref{sec:matrix} we also present the associated matrix representation of size $N=2n+1$. Consequently, the set of exponential cosine-square functions of order $n$ is a special subset of ME($N$) (where $N=2n+1$) which, by construction, is non-negative.
The SCV of an exponential cosine-square function is a complex function of the parameters, whose minimum does not exhibit a closed analytic form. 
That is why we resorted to numerical procedures. 

For a given odd order $N=2n+1$, our goal is to find $\omega$ and $\phi_i$ ($i=1,\ldots,n$) which minimizes the SCV of $f^+(t)$. More precisely, we are looking for efficient numerical methods for finding $\omega$ and $\phi_i$ ($i=1,\ldots,n$) parameters which results in low SCV for high order $N$.  
For efficient numerical minimization of the SCV we need 
\begin{itemize}
\item[$i$)] 
an accurate computation of the SCV based on the parameters with low computational cost 
and
\item[$ii$)] and efficient optimization procedure.
\end{itemize} 
 
In \cite{horvath2016concentrated}, which presents results up to $N=47$, step $i$) is performed by double precision numerical integration, and step $ii$ by the Rechenberg method (discussed in Section \ref{sec:opt}). With that procedure it was not possible to go beyond $N=47$ because of two main reasons. First,  the accuracy of the numerical integration around $N=47$ (step $i$) in \cite{horvath2016concentrated}) decreased to the level that step $ii$ failed to converge properly. 
Apart of the accuracy issues, the computation time of step $i$) increased super-linearly by $N$ and made the computations inhibitive for larger $N$ values.  

In this paper we replace both of these steps of the numerical procedure to increase the accuracy and decrease the computational cost of step $i$) (in Section \ref{sec:scv}) and to find a better minimum in step $ii$) (in Section \ref{sec:opt}). Between these two changes the efficient computation of the objective function (step $i$) plays the more important role, because it makes possible to perform any optimization for larger orders ($N>>47$). 

\section{Efficient computation of the squared coefficient of variation}\label{sec:scv}

To evaluate the objective function of the optimization, namely the SCV, we need efficient methods to compute $\mu_0$, $\mu_1$ and $\mu_2$. 
Deriving the $\mu_i$ parameters based on \eqref{eq:cossquare} is difficult (for large $N$ values). Hence we propose to compute them based on a different representation. 

\subsection{The hyper-trigonometric representation}\label{sec:hypertrig}

The following theorem defines the \emph{hyper-trigonometric form} of the exponential cosine-square functions and provides a recursive procedure to obtain its parameters from $\omega,\phi_i, i=1,\dots, n$.
\begin{theorem}\label{thm:hypertrig}
	An order $N=2n+1$ exponential cosine-square function can be transformed to hyper-trigonometric representation of form
	\begin{align}
	\label{eq:htrig}
	f^{+}(t)=c^{(n)}\cdot e^{-t} + e^{-t}\sum_{k=1}^{n} a^{(n)}_k \cos(k \omega t) + e^{-t}\sum_{k=1}^{n} b^{(n)}_k \sin(k \omega t),
	\end{align}
	where the coefficients $a^{(n)}_k$ and $b^{(n)}_k$, for $1<k<n-1$ and $n>2$, can be determined recursively as
	\begin{align}	
	a^{(n)}_k &= \frac{1}{2}a^{(n-1)}_k + \frac{1}{2}\frac{a^{(n-1)}_{k-1}+a^{(n-1)}_{k+1}}{2}\cos{\phi_n} + \frac{1}{2}\frac{b^{(n-1)}_{k+1}-b^{(n-1)}_{k-1}}{2}\sin{\phi_n}, \label{eq:ank}\\
	b^{(n)}_k &= \frac{1}{2}b^{(n-1)}_k + \frac{1}{2}\frac{b^{(n-1)}_{k-1}+b^{(n-1)}_{k+1}}{2}\cos{\phi_n} + \frac{1}{2}\frac{a^{(n-1)}_{k-1}-a^{(n-1)}_{k+1}}{2}\sin{\phi_n}, \label{eq:bnk}\\
	c^{(n)} &= \frac{1}{2}c^{(n-1)} + \frac{1}{4}a^{(n-1)}_{1}\cos{\phi_n} + \frac{1}{4}b^{(n-1)}_{1}\sin{\phi_n}.\label{eq:cnk}
	\end{align}
	For the boundaries $k=\{1,n-1,n\}, n>2$, the coefficients are
	\begin{align}
	a^{(n)}_{1}&= \frac{1}{2}a^{(n-1)}_{1} + \frac{1}{4}a^{(n-1)}_{2}\cos{\phi_n} + \frac{1}{4}b^{(n-1)}_{2}\sin{\phi_n} + \frac{1}{2}c^{(n-1)}\cos{\phi_n}, \\
	b^{(n)}_{1}&= \frac{1}{2}b^{(n-1)}_{1} + \frac{1}{4}b^{(n-1)}_{2}\cos{\phi_n} - \frac{1}{4}a^{(n-1)}_{2}\sin{\phi_n} + \frac{1}{2}c^{(n-1)}\sin{\phi_n}, \\
	a^{(n)}_{n-1}&= \frac{1}{4}a^{(n-1)}_{n-2}\cos{\phi_n} - \frac{1}{4}b^{(n-1)}_{n-2}\sin{\phi_n} + \frac{1}{2}a^{(n-1)}_{n-1}, \\
	b^{(n)}_{n-1}&= \frac{1}{4}b^{(n-1)}_{n-2}\cos{\phi_n} + \frac{1}{4}a^{(n-1)}_{n-2}\sin{\phi_n} + \frac{1}{2}b^{(n-1)}_{n-1}, \\
	a^{(n)}_{n}&= \frac{1}{4}a^{(n-1)}_{n-1}\cos{\phi_n} - \frac{1}{4}b^{(n-1)}_{n-1}\sin{\phi_n}, \\
	b^{(n)}_{n}&= \frac{1}{4}b^{(n-1)}_{n-1}\cos{\phi_n} + \frac{1}{4}a^{(n-1)}_{n-1}\sin{\phi_n}.
	\end{align}
	The coefficients for $n=1$ are given by
	\begin{align}
	a^{(1)}_{1} = \frac{1}{2}\cos{\phi_1}, ~~~ b^{(1)}_{1} = \frac{1}{2}\sin{\phi_1}, ~~~ c^{(1)} = \frac{1}{2},
	\end{align}
	while for $n=2$ we have
	\begin{align}
	\begin{split}
	a^{(2)}_{1} &= \frac{1}{4}(\cos{\phi_1}+\cos{\phi_2}), ~~~ b^{(2)}_{1} = \frac{1}{4}(\sin{\phi_1}+\sin{\phi_2}), \\
	a^{(2)}_{2} &= \frac{1}{8}\cos(\phi_1-\phi_2), ~~~ b^{(2)}_{2} = \frac{1}{8}\sin(\phi_1+\phi_2),	\\
	c^{(2)} &= \frac{1}{4} + \frac{1}{8}\cos(\phi_1-\phi_2).
	\end{split}
	\end{align}
\end{theorem}
\begin{proof}
The theorem is proved by induction. For the order $N=2n+1$ exponential cosine-square function we introduce the notation $f^{(n)}(t)$. Assuming that the theorem holds for $n-1$, we are going to show that it holds for $n$ as well. We have that 
\begin{align*}
\begin{split}
f^{(n)}(t) &= f^{(n-1)}(t)\cos^2\left(\frac{\omega t - \phi_i}{2}\right) = f^{(n-1)}(t)\left(\frac{1}{2}+\frac{1}{2}\cos(\omega t-\phi_n)\right) \\
&=f^{(n-1)}(t)/2+f^{(n-1)}(t)\cos(\omega t) \cos(\phi_n)/2 + f^{(n-1)}(t)\sin(\omega t) \sin(\phi_n)/2 \\
&=c^{(n-1)}\cdot e^{-t}/2 + e^{-t}\sum_{k=1}^{n-1} a^{(n-1)}_k \cos(k \omega t)/2 + e^{-t}\sum_{k=1}^{n-1} b^{(n-1)}_k \sin(k \omega t)/2 \\
&\quad+c^{(n-1)} e^{-t}\cos(\omega t) \cos(\phi_n)/2 + c^{(n-1)} e^{-t}\sin(\omega t) \sin(\phi_n)/2  \\
&\quad+ e^{-t}\sum_{k=1}^{n-1}a^{(n-1)}_k \cos(k \omega t) \cos(\omega t) \cos(\phi_n)/2 \\
&\quad+e^{-t}\sum_{k=1}^{n-1}a^{(n-1)}_k \cos(k \omega t)\sin(\omega t) \sin(\phi_n)/2 \\
&\quad+e^{-t}\sum_{k=1}^{n-1}b^{(n-1)}_k \sin(k \omega t)\cos(\omega t) \cos(\phi_n)/2 \\
&\quad+e^{-t}\sum_{k=1}^{n-1}b^{(n-1)}_k \sin(k \omega t)\sin(\omega t) \sin(\phi_n)/2.
\end{split}
\end{align*}
Using the identities
\begin{align*}
\cos(k \omega t) \cos(\omega t) &= \cos((k+1)\omega t)/2 + \cos((k-1)\omega t)/2, \\
\sin(k \omega t) \sin(\omega t) &= -\cos((k+1)\omega t)/2 + \cos((k-1)\omega t)/2, \\
\cos(k \omega t) \sin(\omega t) &= \sin((k+1)\omega t)/2 - \sin((k-1)\omega t)/2, \\
\sin(k \omega t) \cos(\omega t) &= \sin((k+1)\omega t)/2 + \sin((k-1)\omega t)/2,
\end{align*}
and collecting the coefficients corresponding to $\cos(k \omega t)$ and $\sin(k \omega t)$ provides \eqref{eq:ank} and \eqref{eq:bnk}. Terms from $\cos((k-1)\omega t)$ at $k=1$ contribute to $c^{(n)}$, leading to \eqref{eq:cnk}. The relations for the boundaries can be obtained in a similar manner.
\end{proof}

The hyper-trigonometric representation makes it possible to express the Laplace transform (LT) and the $\mu_i$ moments in a simple and compact way.

\begin{corollary}
\label{cor:moments}	
The LT and the $\mu_i, i=0,1,2$ moments of the exponential cosine-square function are given by 
\begin{align}
f^{*}(s)=\int_{0}^\infty e^{-st}f^{+}(t)\,dt=\frac{c^{(n)}}{1+s} + \sum_{k=1}^n \frac{a_k^{(n)}(1+s)+b_k^{(n)}k\omega}{(1+s)^2+(k\omega)^2},
\end{align}
and
\begin{align}
\begin{split}
\mu_0 &= c^{(n)} + \sum_{k=1}^n \frac{a_k^{(n)}+b_k^{(n)}k\omega}{1+(k\omega)^2},  \\
\mu_1 &= c^{(n)} + \sum_{k=1}^n \frac{ a_k^{(n)}+2b_k^{(n)}k\omega-a_k^{(n)}(k\omega)^2}{(1+(k\omega)^2)^2}, \\
\mu_2 &= 2 c^{(n)} + \sum_{k=1}^n \frac{2 a_k^{(n)} + 6b_k^{(n)}k\omega - 6 a_k^{(n)}(k\omega)^2 - 2b_k^{(n)}(k\omega)^3}{(1+(k\omega)^2)^3}. \label{eq:moments}
\end{split}
\end{align}
\end{corollary}

\begin{proof}
Since \eqref{eq:htrig}, is linear, it can be Laplace transformed term-by-term using the following relations
\begin{align*}
LT(e^{-t})  &= \frac{1}{s+1}, \\ 
LT(e^{-t} \cos(\omega t)) &= \frac{s+1}{(s+1)^2+\omega^2} , \\
LT(e^{-t} \sin(\omega t)) &= \frac{\omega}{(s+1)^2+\omega^2}. 
\end{align*}
Based on $f^{*}(s)$, the $\mu_i$ moments can be computed using the LT moment relation 
\[ \mu_i
= (-1)^i \frac{d^i}{ds^i} f^{*}(s)\bigg|_{s=0}. \]
\end{proof}

In order to compute the matrix representation of the exponential cosine-square function, we introduce one more representation and the associated transformations. 

\subsection{The spectral representation}\label{sec:spectral}

The hyper-trigonometric representation makes it easy to obtain the spectral form of $f^{+}(t)$ as
\begin{align}
\begin{split}
f^{+}(t) =& ~c^{(n)}e^{-t} \\
&~ + \frac{1}{2}\sum_{k=1}^n \left(a_k^{(n)}+\ii b_k^{(n)}\right)e^{-(1+\ii k\omega)t} + 
\left(a_k^{(n)}-\ii b_k^{(n)}\right)e^{-(1-\ii k\omega)t},\label{eq:spectral}
\end{split}
\end{align}
where $\ii$ denotes the imaginary unit. 
In this form, $f^{+}(t)$ is a sum of exponential functions with complex exponential coefficients.  Specifically, it has a real exponential coefficient $-1$ and $n$ pairs of complex conjugate coefficients $-1\pm \ii k\omega$ ($k=1,\ldots,n$). 
We refer to  this representation as \emph{spectral representation} because these exponential coefficients are the eigenvalues of the matrix representation. 

The constant multipliers of the exponential functions with complex conjugate exponential coefficients are complex conjugates as well, which ensures that $f^{+}(t)$ is real. 
The real and the imaginary parts of the constant multipliers are given by $a_k^{(n)}$ and $b_k^{(n)}$, respectively.

\subsection{The matrix representation}\label{sec:matrix}

\begin{corollary}
For $\forall t\geq 0$, the $f^{+}(t) = e^{-t} \prod_{i=1}^{n} \cos^2\left(\frac{\omega t - \phi_i}{2}\right)$
function has a size $N=2n+1$ matrix representation $f^+(t)= \ve{\beta} e^{\mx{B}t} (-\mx{B})\1$,
where row vector $\ve{\beta}$ is composed by the elements
\begin{align}
\beta_1&=c^{(n)}, \\
\beta_{2k} &= \frac{1}{2}\frac{a_k^{(n)}(1+k\omega)-b_k^{(n)}(1-k\omega)}{1+(k\omega)^2}, \\
\beta_{2k+1} &= \frac{1}{2}\frac{a_k^{(n)}(1-k\omega)+b_k^{(n)}(1+k\omega)}{1+(k\omega)^2},
\end{align}
for $k=1,\dots,n$, and matrix $\mx{B}$ is given by
\begin{align}
\mx{B}=\begin{bmatrix}
-1 \\
& -1 & -\omega \\
& \omega & -1 \\
&&& -1 & -2\omega \\
&&& 2\omega & -1 \\
&&&&\ddots & \ddots& \ddots\\
&&&&&& -1 & -n\omega \\
&&&&&& n\omega & -1 \\
\end{bmatrix}~.
\end{align}
\end{corollary}

\begin{proof}
Matrix $\mx{B}$ is block diagonal, hence its eigenvalues are the eigenvalues of the diagonal blocks. The eigenvalue associated with the first diagonal block of size $1$ is $-1$, the rest of the diagonal blocks are of size $2$. The eigenvalues associated with the $k$th size $2$ diagonal block are $-1\pm\ii k\omega$. 

Due to the block diagonal structure of $\mx{B}$ the diagonal blocks
and their associated multipliers in vector $\ve{\beta}$ represent the terms of the hyper-trigonometric representation in \eqref{eq:htrig} one by one. 
The first diagonal block of size $1$ and its associated multiplier $\beta_1=c^{(n)}$
represent the $c^{(n)} e^{-t}$ term of \eqref{eq:htrig}. 
The $k$th size $2$ diagonal block
and its associated multipliers $\beta_{2k}$ and $\beta_{2k+1}$ represent the $k$th term of  summation \eqref{eq:htrig}, that is 
\begin{align*}
&[\beta_{2k}, \beta_{2k+1}] ~ e^{\begin{small}
	\begin{array}{|cc|}
\hline
-1 & -k\omega \\
k\omega & -1 \\
\hline
\end{array}
	\end{small}~t}
~~ 
\begin{array}{|c|}
\hline
1+ k\omega \\
1- k\omega  \\
\hline
\end{array} 
= e^{-t} (a^{(n)}_k \cos(k \omega t) + b^{(n)}_k \sin(k \omega t)), 
\end{align*}  
which is based on the expansion 
\[  e^{\begin{small}
	\begin{array}{|cc|}
	\hline
	-1 & -k\omega \\
	k\omega & -1 \\
	\hline
	\end{array}
	\end{small}~t} = 
\begin{array}{|cc|}
\hline
e^{-t} \cos(k \omega t) & -e^{-t}\sin(k \omega t) \\
e^{-t} \sin(k \omega t) & e^{-t}\cos(k \omega t) \\
\hline
\end{array}. 
\]
\end{proof}

\subsection{Numerical computation of the moments}

Theorem \ref{thm:hypertrig} together with Corollary \ref{cor:moments} provide a very efficient explicit method to compute the SCV based on the parameters $\omega,\phi_i, i=1,\dots, n$. 

There is one numerical issue that has to be taken care of when applying this numerical procedure with floating point arithmetic for large $n$ values. 
To evaluate the SCV, coefficients $a_k^{(n)},b_k^{(n)},c^{(n)}$ need to be obtained from the  $\omega$ and $\phi_i, i=1,\dots,n$ parameters. The recursion defined in Theorem \ref{thm:hypertrig} involves multiplications between bounded numbers (sine and cosine always fall into $[-1,+1]$), which is beneficial from the numerical stability point of view, but subtractions are unfortunately also present, leading to loss of precision. To overcome this loss of precision, we introduced increased precision floating point arithmetic 
both in our Mathematica and C++ implementations\footnote{In C++ we used to \texttt{mpfr} library for multi-precision computations}. Mathematica can quantify the precision loss, enabling us to investigate this issue experimentally. According to Figure \ref{fig:precloss}, the number of accurate decimal digits lost when evaluating the SCV from the $\omega,\phi_i$ parameters (computed by the \texttt{Precision} function of Mathematica), denoted by $L_n$, is nearly linear and can be approximated by 
\begin{align}
L_n \approx 1.487 + 0.647 n. \label{eq:ln}
\end{align}
In the forthcoming numerical experiments we have set the floating point precision to $L_n + 16$ decimal digits to obtain results up to $16$ accurate decimal digits, and this precision setting eliminated all numerical issues.

\begin{figure}
	\centering
	\resizebox{0.695\textwidth}{!}{
		\begin{tikzpicture}
		\tikzstyle{every node}=[font=\scriptsize]
		\tikzstyle{every axis plot}=[semithick]
		\begin{axis}[
			xlabel={$n=(N-1)/2$},
			ylabel={$L_n$, lost decimal digits},
			domain=1:74,
			width=7cm,			
			height=5cm,
			legend pos=north west,			
			cycle list name=linestyles*
			]
			\addplot[only marks,mark=+,mark size=1pt] table[x index=0,y index=1] {./precloss.txt};
			\addlegendentry{Experimental}
			\addplot[solid] {1.487+x*0.647};
			\addlegendentry{Approximated}
		\end{axis}
	\end{tikzpicture}
	}	
	\caption{The precision loss while computing the SCV}
	\label{fig:precloss}
\end{figure}
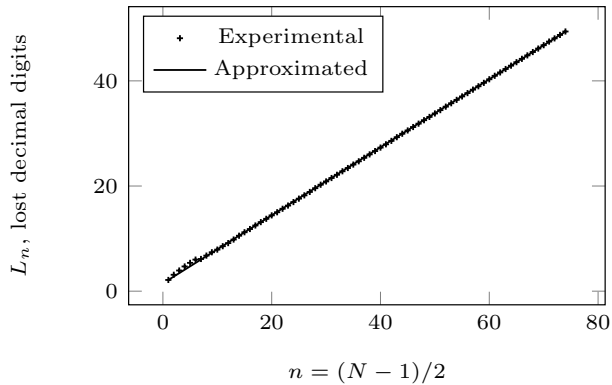

It is important to note that the high precision is needed only to calculate the $a_k^{(n)},b_k^{(n)},c^{(n)}$ coefficients and the SCV itself. Representing parameters $\omega,\phi_i$ themselves does not need extra precision, and the resulting exponential cosine-square function $f(t)$ can be evaluated with machine precision as well (in the range of our interest, $n\leq 1000$).

A basic pseudo-code of the computation of the SCV with the indications where high precision is needed is provided by Algorithm \ref{alg:scv}.

\begin{algorithm}[H]
	\begin{algorithmic}[1]
		\Procedure{ComputeSCV}{$n,\omega,\phi_i,i=1,\dots,n$}
		
		\State Compute the required precision, $L_n$, from \eqref{eq:ln}
		\State Convert $\omega,\phi_i,i=1,\dots,n$ to $L_n+16$ digits precision
		\State Calculate $a_k^{(n)},b_k^{(n)},c^{(n)}, k=1,\dots,n$, recursively by Theorem \ref{thm:hypertrig} (\textbf{high precision})
		\State Calculate moments $\mu_0,\mu_1,\mu_2$  according to \eqref{eq:moments} (\textbf{high precision})
		\State Calculate $SCV=\frac{\mu_0\mu_2}{\mu_1^2}-1$ (\textbf{high precision})
		\State Convert $SCV$ to machine precision
		\State \Return $SCV$
		\EndProcedure
	\end{algorithmic}
	\caption{Pseudo-code for the computation of the SCV}\label{alg:scv}
\end{algorithm}

\section{Minimizing the squared coefficient of variation}
\label{sec:opt}

Given the size of the representation $N=2n+1$, the $f^+(t)$ function providing the minimal SCV is obtained by minimizing \eqref{eq:SCV} subject to $\omega$ and $\phi_i, i=1,\dots,n$. The form of SCV does not allow a symbolic solution, and its numerical optimization is challenging, too. The surface to optimize has many local optima, hence simple gradient descent procedures failed to find the global optimum and are sensitive to the initial guess.

\subsection{Optimizing the parameters}

In the numerical optimization of the parameters, we had success with evolutionary optimization methods, in particular with \emph{evolution strategies}. The results introduced in \cite{horvath2016concentrated} were obtained by one of the simplest evolution strategy, the Rechenberg method \cite{rechenberg1978evolutionsstrategien}. In \cite{horvath2016concentrated}, it was the high computational demand of the numerical integration needed to obtain the SCV and its reduced accuracy that prevented the optimization for $N>47$ ($n>23$). 

However, computing the SCV based on the hyper-trigonometric representation using the results of Section \ref{sec:hypertrig} allows us to evaluate the moments orders of magnitudes faster and more accurately, enabling the optimization for higher $n$ values. With the Rechenberg method (\cite{rechenberg1978evolutionsstrategien}, also referred to as (1+1)-ES in the literature) it is possible to obtain low SCV values relatively quickly for orders as high as $n=125$,
but these values are suboptimal in the majority of the cases.  

With more advanced evolution strategies the optimal SCV can be approached better. Our implementation supports the covariance matrix adoption evolution strategy (CMA-ES \cite{hansen2006cma}), and one of its variant, the BIPOP-CMA-ES with restarts \cite{hansen2009benchmarking}. Starting from a random initial guess, we got very low SCV much quicker with the CMA-ES then with the (1+1)-ES with similar suboptimal minimum values (cf. Fig. \ref{fig:mincv_heur_largen}). The limit of applicability of CMA-ES is about $n=180$. The best solution (lowest SCV for the given order), however, was always provided by the BIPOP-CMA-ES method, although it is by far the slowest among the three methods we studied. 
In fact, we believe that BIPOP-CMA-ES returned the global optimum for $n=1,\dots,74$,
and we investigate the properties of those solutions in the next sections as they were the optimal ones.
For $n>74$, we can still compute low SCV functions with the BIPOP-CMA-ES method, but its computation time gets to be prohibitive, and we are less confident about the global minimality of the results.

For our particular problem, the running time, $T$, and the quality of the minimum, $Q$, obtained by the different optimization methods can be summarized as follows
\begin{align*}
T_{\text{CMA-ES}} < &~ T_{\text{(1+1)-ES}} <<  T_{\text{BIPOP-CMA-ES}}, \\
Q_{\text{CMA-ES}} \sim &~ Q_{\text{(1+1)-ES}} <  Q_{\text{BIPOP-CMA-ES}}. 
\end{align*}

\subsection{Properties of the minimal SCV solutions}

\begin{figure}
	\centering
	\resizebox{0.695\textwidth}{!}{
		\begin{tikzpicture}
		\tikzstyle{every axis plot}=[semithick]
		\begin{axis}[
		xlabel={$n=(N-1)/2$},
		ylabel={$SCV_n$},
		xmode=log,
		ymode=log,
		domain=1:74,
		width=8cm,
		height=6cm,
		cycle list name=linestyles*,
		legend pos=north east
		]
		\addplot table[x index=0,y index=1] {./cv2.txt};
		\addlegendentry {ME($N$)}			
		\addplot {1/(2*x+1)};
		\addlegendentry {$1/N$}			
		\addplot {2/(2*x+1)^2};
		\addlegendentry {$2/N^2$}			
		\end{axis}
		\end{tikzpicture}
	}
	\caption{The minimal SCV of the exponential cosine-square functions as the function of $n$ 
		in log-log scale}
	\label{fig:mincv}
\end{figure}
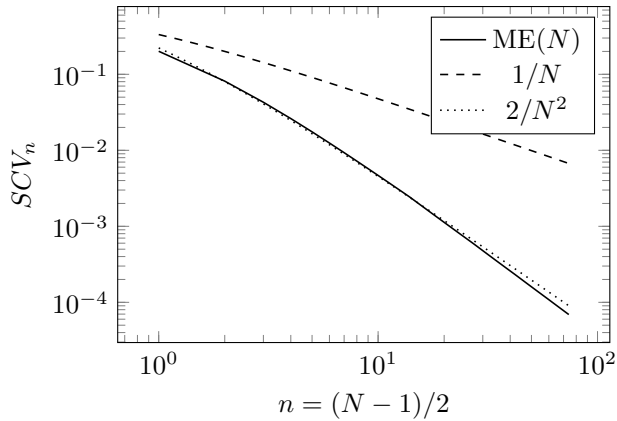

The minimal SCV values obtained by the BIPOP-CMA-ES optimization, which we consider as being optimal for $n=1,\dots,74$, are depicted in Figure \ref{fig:mincv}.
Apart from the minimal SCV values of the exponential cosine-square functions, Figure \ref{fig:mincv} also plots $1/N$ and $2/N^2$, for comparison. 
The $SCV=1/N$ is known to be the minimal SCV value of phase-type (PH) distributions of order $N$ \cite{[ALDO87]}, which form a subset in the set of ME distributions by assuming positive off-diagonal elements for $\mx{B}$ and positive elements for $\ve{\beta}$.
The  $2/N^2$ curve is reported to be the approximate decay rate in \cite{horvath2016concentrated}, up to $n=23$ ($N=47$). 

Figure \ref{fig:mincv} indicates that the SCV decreases much faster than $1/N$ and a bit faster than $2/N^2$. Indeed, $2/N^2$ is a good approximate up to $n=23$, but the decay seems to decrease below $2/N^2$ for $n>23$. We suspect that the decrease is asymptotically polynomial (at $n\rightarrow\infty$), that we checked by plotting the $SCV n^{2.14}$ function in Figure \ref{fig:faster}. While the exponent is determined empirically and might be slightly off, it is visible that the convergence is faster than $1/n^2$.

\begin{figure}
	\centering
	\resizebox{0.495\textwidth}{!}{
		\begin{tikzpicture}
		\tikzstyle{every node}=[font=\scriptsize]
		\tikzstyle{every axis plot}=[semithick]
		\begin{axis}[
			xlabel={$n=(N-1)/2$},
			ylabel={$SCV_n \cdot n^{2.14}$},
			ymode=log,
			width=7cm,
			height=5cm,
			cycle list name=linestyles*
			]
			\addplot table[x index=0,y expr=\thisrowno{1}*\thisrowno{0}^2.14] {./cv2.txt};
		\end{axis}
	\end{tikzpicture}
	}	

	\caption{The function $S(n)=SCV_n\cdot n^{2.14}$ with logarithmic $y$ axis}
	\label{fig:faster}
\end{figure}
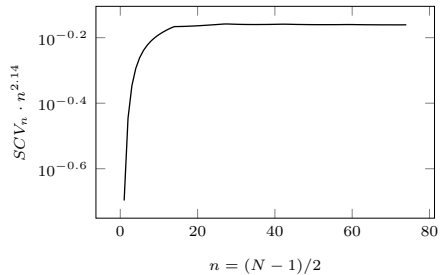

\subsection{The parameters providing the minimal SCV}

\begin{figure}
	\centering
	\resizebox{0.495\textwidth}{!}{
	\begin{tikzpicture}
	\tikzstyle{every node}=[font=\scriptsize]
	\tikzstyle{every axis plot}=[semithick]
	\begin{axis}[
	xlabel={$n=(N-1)/2$},
	ylabel={$\omega$},
	ymin=0,
	width=7cm,
	height=5cm,
	cycle list name=linestyles*
	]
	\addplot table[x index=0,y index=1] {./omega.txt};
	\end{axis}
	\end{tikzpicture}}
\hfill	
	\resizebox{0.495\textwidth}{!}{
	\begin{tikzpicture}
	\tikzstyle{every node}=[font=\scriptsize]
	\tikzstyle{every axis plot}=[semithick]
	\begin{axis}[
	xlabel={$n=(N-1)/2$},
	ylabel={$\omega$},
	xmode=log,
	ymode=log,
	ymin=0,
	width=7cm,
	height=5cm,
	cycle list name=linestyles*
	]
	\addplot table[x index=0,y index=1] {./omega.txt};
	\end{axis}
	\end{tikzpicture}}
	\caption{The $\omega$ parameter providing the minimal SCV in lin-lin and log-log scales}
	\label{fig:omega}
\end{figure}
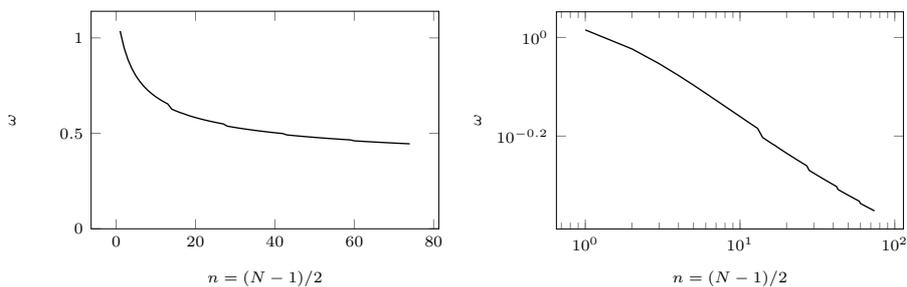

In this section, we investigate the parameters 
corresponding to the minimal SCV and provide an intuitive explanation for the observed behaviour. 
Figure \ref{fig:omega} depicts the optimal $\omega$ parameter as the function of $n$. It shows slow decrease with some inhomogeneity around $n=14,28,43,60$. 
Figure \ref{fig:omega} suggests that $\omega$ tends to $0$ as $n\rightarrow 0$, and the inhomogeneity is related with the behaviour of $\phi_k$ parameters, as it is detailed below. 

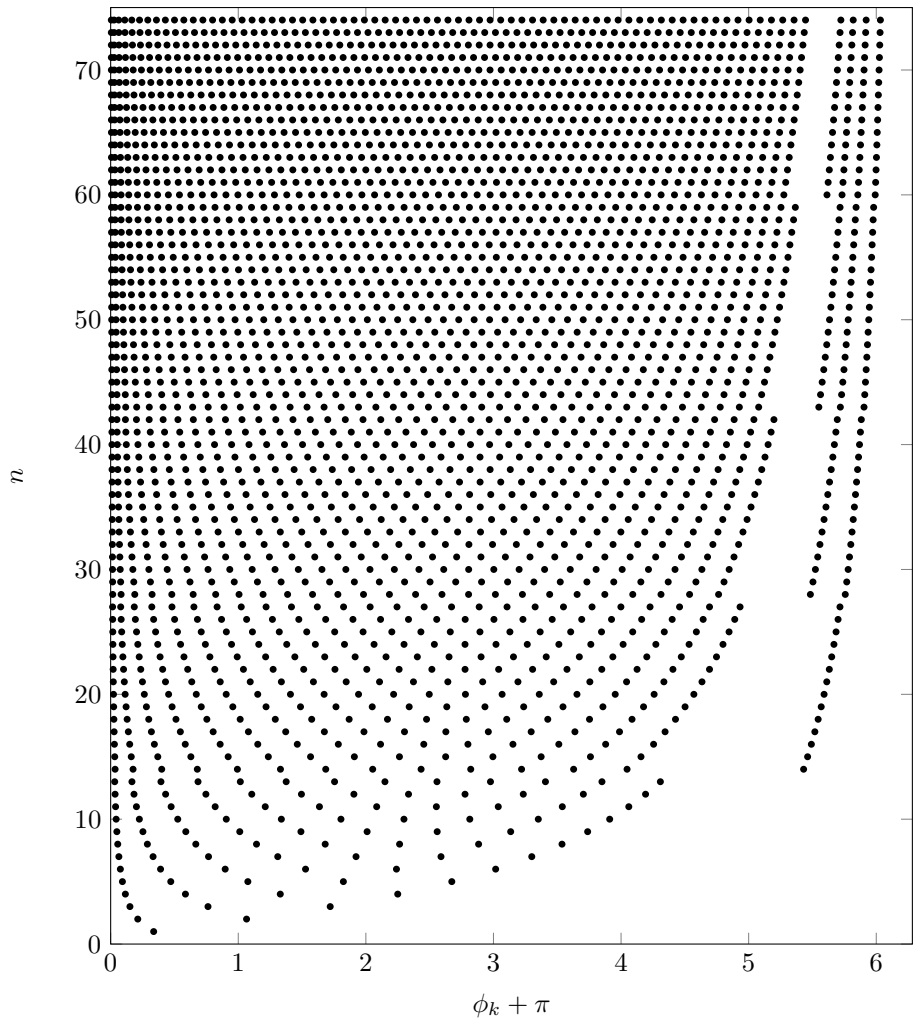
\begin{figure}
	\centering
	\begin{tikzpicture}
	\tikzstyle{every axis plot}=[semithick]
	\begin{axis}[
	xlabel={$\phi_k+\pi$},
	ylabel={$n$},
	ymin=0,
	ymax=75,
	xmin=0,
	xmax=6.283,
	width=\textwidth,
	height=14cm
	]
	\addplot[only marks,mark=*,mark size=1pt] table[x index=0,y index=1] {./phi.txt};
	\end{axis}
	\end{tikzpicture}	
	\caption{The location of the $\phi_k,k=1,\dots,n$ parameters providing the minimal SCV}
	\label{fig:phi}
\end{figure}

The visual appearance of the optimal $\phi_k$ parameters in Figure \ref{fig:phi} reveal many more interesting properties. First of all, since the period of the cosine-square function is $\pi$, the $\phi_k$ parameters in \eqref{eq:cossquare} are $2\pi$ periodic. That is, adding an integer times $2\pi$ to any of the $\phi_k$ parameters does not change the $f^+(t)$ function. In Figure \ref{fig:phi} we transformed all $\phi_k$ parameters to the $(-\pi,\pi)$ range and depicted $\phi_k+\pi$ instead of $\phi_k$ because $\phi_k+\pi$ indicates the location where the cosine-square term with $\phi_k$ gets to be zero in the $(0,2\pi)$ cycle, which is the $(0,2\pi/\omega)$ interval of $f^+(t)$. The $n$th row in Figure \ref{fig:phi} depicts $n$ points, which are the optimal $\phi_k+\pi$ values for $k=1,\ldots,n$. In these $n$ points of the $(0,2\pi)$ cycle $f^+(t)$ equals zero. In between these zeros 
$f^+(t)$ has humps. 

Figure \ref{fig:cossquare} demonstrates the effect of $\phi_k$ on $f^+(t)$ for $\omega=1$, $\phi_i+\pi=0.1, 1, 2$ for $i=1,2,3$. In the $(0,2\pi)$ interval, both, $f^+(t)$ and $ \prod_{i=1}^{n} \cos^2\left(\frac{ t - \phi_i}{2}\right)$, have zeros at $0.1$, $1$ and $2$. The left curve depicts $f^+(t)$ which contains the  exponential attenuation while the right curve depicts $\prod_{i=1}^{n} \cos^2\left(\frac{ t - \phi_i}{2}\right)$, which is without the exponential attenuation. In the left curve the sizes of the humps depend on the distance of the neighbouring zeros. The hump between $(0.1,1)$ is smaller than the one between $(1,2)$, which indicates that the closer the neighbouring zeros are the smaller the humps are.  
In the left curve, the exponential attenuation also affects the sizes of the humps. 
With the exponential attenuation the hump between $(0.1,1)$ is larger than the one between $(1,2)$. Additionally, the SCV is more sensitive to the humps farther from the main peak, which motivates the fact that the $\phi_k+\pi$ parameters are more concentrated around $0$ to make the function as flat as possible near $t=0$, where the exponential attenuation is rather weak. 
In Figure \ref{fig:cossquare}, right to the main peak $f^+(t)$ is suppressed by 
the exponential attenuation. 

In Figure \ref{fig:phi}, nor $n<14$ the zeros are located between $0$ and $5$, which means that in this range of $n$ the $f^+(t)$ is kept close to zero by the cosine-square functions in interval $(0,5)$, it has a peak between $5$ and $2\pi$, and the next cycles are suppressed by the exponential attenuation. 

\begin{figure}
	\centering
	\resizebox{0.95\textwidth}{!}{
		\begin{tikzpicture}
		\tikzstyle{every axis plot}=[semithick]
		\begin{axis}[
			yticklabel style={
				/pgf/number format/fixed,
				/pgf/number format/precision=5
			},
			scaled y ticks=false,
			xlabel={$t$},
			ylabel={$f^+(t)$},
			domain=0:6.7,
			samples=200,
			ymin=0,
			xmin=0,
			xmax=6.7,
			width=10cm,			
			height=5cm,
			legend pos=north west,			
			cycle list name=linestyles*
			]
			\addplot[solid]
			{exp(-x)*cos(deg((x-3.24)/2))^2*cos(deg((x-4.14)/2))^2*cos(deg((x-5.14)/2))^2};
		\end{axis}
		\end{tikzpicture}
		$~~~$
		\begin{tikzpicture}
		\tikzstyle{every axis plot}=[semithick]
		\begin{axis}[
			yticklabel style={
				/pgf/number format/fixed,
				/pgf/number format/precision=5
			},
			scaled y ticks=false,
			xlabel={$t$},
			ylabel={$\prod_{i=1}^{n} \cos^2\left(\frac{ t - \phi_i}{2}\right)$},
			domain=0:2.26,
			samples=200,
			ymin=0,
			ymax=0.005,
			xmin=0,
			xmax=2.26,
			width=10cm,			
			height=5cm,
			legend pos=north west,			
			cycle list name=linestyles*
			]
			\addplot[solid]{cos(deg((x-3.24)/2))^2*cos(deg((x-4.14)/2))^2*cos(deg((x-5.14)/2))^2};
		\end{axis}
		\end{tikzpicture}
	}	
	\caption{$f^+(t)$ and $ \prod_{i=1}^{n} \cos^2\left(\frac{\omega t - \phi_i}{2}\right)$, with $\omega=1$, $\phi_i+\pi=0.1, 1, 2$ for $i=1,2,3$}
	\label{fig:cossquare}
\end{figure}
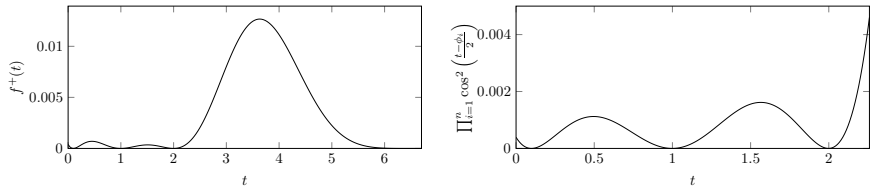

For $n\geq14$, there is a gap in the sequence of $\phi_k$ parameters, indicating the location of a spike. It means that the exponential attenuation is not strong enough for suppressing $f^+(t)$ right after the spike, and the minimal SCV is obtained when some cosine-square terms are used to enforce the suppression beyond the spike.

The number of cosine-square terms used to suppress $f^+(t)$ beyond the spike changes at $n=14,28,43,60$. These changes result in the small inhomogeneity in the $\omega$ values at 
$n=14,28,43,60$ in Figure \ref{fig:omega}.

According to our intuitive understanding $\omega$ tends to $0$ as $n\rightarrow \infty$,
because the cosine squared terms are more efficient to suppress $f^+(t)$ than the exponential attenuation, and for large $n$ values, the cosine squared terms create a sharp spike inside the $0,2\pi$ cycle (which is the $(0,2\pi/\omega)$ interval), such that zeros are located on both sides of the spike. 
The cosine squared terms are $2\pi/\omega$ periodic and consequently, there is a spike also in the $(2\pi/\omega,4\pi/\omega)$ interval which is suppressed by exponential attenuation. 
The exponential attenuation in a $2\pi/\omega$ long interval is $e^{2\pi/\omega}$. In order to efficiently suppress the spike in the $(2\pi/\omega,4\pi/\omega)$ interval, $\omega$ has to be small. 

\section{Heuristic optimization with 3 parameters}
\label{sec:heuristic}

According to the previously discussed approach the number of parameters to optimize increases with $n$. This drawback limits the applicability of the general optimization procedures to about $n\leq 74$ in case of BIPOP-CMA-ES and about $n\leq 180$ in case of the basic CMA-ES. By these $n$ values the optimization procedure takes several days to terminate on our average PC clocked at 3.4 GHz.

While the $f^+(t)$ functions obtained this way have an extremely low ($\approx 10^{-5}$) SCV already, some applications might benefit from functions with even lower SCV. To overcome this limitation we developed a suboptimal heuristic procedure, that aims to obtain low SCV for a given large order $n$. The main idea behind this procedure is that we exploit the regular pattern of the optimal $\phi_i$ parameters instead of optimizing each of them individually.

The appropriate description of the $\phi_i+\pi$ parameters (see Figure \ref{fig:phi}) is, however, not obvious. We have tried to spread these parameters evenly in $(0,2\pi)$, and we also tried to approximate the increased density near zero. These initial experiments did not lead to satisfactory results.

The heuristic procedure that met our expectations has to optimize only three parameters, independent of the order $n$.  The procedure is based on the assumption that the location of the spike plays the most important role in the SCV, and the exact location of the zeros of the cosine-squared terms are less important, the only important feature is that the cosine-squared terms should suppress $f^+(t)$ uniformly in the 
$(0,2\pi)$ cycle -- apart of the spike. 

Based on this assumption we set the zeros of the cosine-squared terms equidistant. 
This way the place of the spike ($p$) and its width ($w$) inside the $(0,2\pi)$ interval completely define the $\phi_k$ values for a given order $n$. 

The distance of the zeros ($d$) and the number of zeros before the spike ($i$) are
\begin{align} 
d = \frac{2\pi-w}{n}, ~~~ i = \left\lfloor \frac{p-w/2}{d} +\frac{1}{2} \right\rfloor,  \label{eq:distzero}
\end{align}
and for $k=1,\ldots,n$ the zeros are located at
\begin{align}
\phi_k+\pi = \left\{ \begin{array}{ll}
(k-1/2) d & \mbox{if~} k\leq i,\\
(k-1/2) d + w & \mbox{if~} k > i.
\end{array}
\right.\label{eq:phiheur}
\end{align}

\begin{figure}
	\centering
	\resizebox{0.695\textwidth}{!}{
		\begin{tikzpicture}
		\tikzstyle{every axis plot}=[semithick]
		\begin{axis}[
		log ticks with fixed point,
		xtick={10, 20, 50, 75},
		xlabel={$n=(N-1)/2$},
		ylabel={$SCV_n$},
		xmode=log,
		ymode=log,
		xmin=10, xmax=75,
		width=8cm,
		height=6cm,
		cycle list name=linestyles*,
		legend pos=north east
		]
		\addplot table[x index=0,y index=1] {./cv2.txt};
		\addlegendentry {BIPOP-CMA-ES}			
		\addplot table[x index=0,y index=1] {./cv2_heur.txt};
		\addlegendentry {3-parameter heuristic}			
		\end{axis}
		\end{tikzpicture}
	}
	\caption{The minimal and the heuristic SCV as a function of order $n$ in log-log scale}
	\label{fig:mincv_heur}
\end{figure}

The obtained heuristic procedure has only 3 parameters to optimize: $\omega$, $p$ and $w$ (see Algorithm \ref{alg:scvheur}). 
The SCV values computed by this heuristic optimization procedure are depicted in Figure \ref{fig:mincv_heur} for $n=10,\ldots,74$. The figure indicates that
for small $n$ values ($n<15$)  the procedure is inaccurate, 
but it is not a problem because the minimal SCV can be computed quickly in these cases. 
For larger $n$ values ($n\geq 15$) 
the SCV provided by the heuristic procedure is less than the double of the minimal SCV in the given range.  
Assuming that this ratio to the optimal SCV remains valid also for $n>74$ the heuristic procedure which is applicable up to $n=1000$, is an efficient tool to compute highly concentrated non-negative matrix exponential functions for large order $n$ values. 

\begin{algorithm}[H]
	\begin{algorithmic}[1]
		\Procedure{ComputeSCV-heuristic}{$\omega,p,w$}
	
		\State Obtain the distance of zeros $d$ and threshold $i$ by \eqref{eq:distzero}		
		\State Obtain $\phi_i$ for $i=1,\dots,n$ by \eqref{eq:phiheur}
		\State Compute $SCV$ by Algorithm \ref{alg:scv}
		\State \Return $SCV$
		\EndProcedure
	\end{algorithmic}
	\caption{The objective function of the heuristic method}\label{alg:scvheur}
\end{algorithm}

Figure \ref{fig:mincv_heur_largen} depicts the SCV obtained by the heuristic procedure for large $n$ values, compared with the outputs of the highly accurate BIPOP-CMA-ES and the faster CMA-ES optimization procedures.
Figure \ref{fig:mincv_heur_largen} suggests that the heuristic optimizations remains very close to the minimum also for larger $n$ values and the SCV obtained by the heuristic optimization maintains its polynomial decay between $n^{-2.1}$ and $n^{-2.2}$.

\begin{figure}
	\centering
	\resizebox{0.695\textwidth}{!}{
		\begin{tikzpicture}
		\tikzstyle{every axis plot}=[semithick]
		\begin{axis}[
		/pgf/number format/.cd, 1000 sep={},
		log ticks with fixed point,
		xlabel={$n=(N-1)/2$},
		ylabel={$SCV_n$},
		xtick={10, 20, 50, 100, 200, 500, 1000},
		xmode=log,
		ymode=log,
		xmin=10, 
		width=8cm,
		height=6cm,
		cycle list name=linestyles*,
		legend pos=north east
		]
		\addplot table[x index=0,y index=1] {./cv2.txt};
		\addlegendentry {BIPOP-CMA-ES}			
		\addplot table[x index=0,y index=1] {./cv2_heur.txt};
		\addlegendentry {3-parameter heuristic}			
		\addplot table[x index=0,y index=1] {./cv2_cmaes.txt};
		\addlegendentry {CMA-ES}			
		\end{axis}
		\end{tikzpicture}
	}
	\caption{The minimal and the heuristic SCV as a function of order $n$ in log-log scale}
	\label{fig:mincv_heur_largen}
\end{figure}

\section{Conclusion}\label{sec:concl}

The hyper-trigonometric representation introduced in this paper enables the efficient, explicit computation of the squared coefficient of variation of exponential cosine-squared distributions. On the top of this result, by selecting the appropriate numerical precision and  picking a suitable numerical optimization method, we managed to create matrix exponential functions with extremely low SCV. Such non-negative, low-SCV functions are important ingredients of several numerical procedures, including the numerical inverse Laplace transform methods that are widely used in many research fields.

\bibliographystyle{plain}
\bibliography{mefication}

\end{document}